\documentclass[journal]{IEEEtran}
\usepackage{graphicx}
\usepackage{cite}
\usepackage{amsmath}
\usepackage{amsfonts}
\usepackage{amssymb}
\usepackage{algorithm}
\usepackage{algorithmic}

\usepackage{psfrag}
\usepackage{epsfig}
\usepackage{multicol}

\newtheorem{theorem}{Theorem}
\newtheorem{proposition}[theorem]{Proposition}

\newenvironment{proof}[1][Proof]{\begin{trivlist}
\item[\hskip \labelsep {\bfseries #1}]}{\end{trivlist}}

\newcommand{\qed}{\nobreak \ifvmode \relax \else
      \ifdim\lastskip<1.5em \hskip-\lastskip
      \hskip1.5em plus0em minus0.5em \fi \nobreak
      \vrule height0.75em width0.5em depth0.25em\fi}

\newcommand{\bm}{\mathbf}
\newcommand{\be}{\begin{equation}}
\newcommand{\ee}{\end{equation}}
\newcommand{\bse}{\begin{subequations}}
\newcommand{\ese}{\end{subequations}}
\newcommand{\bea}{\begin{eqnarray}}
\newcommand{\eea}{\end{eqnarray}}

\newcommand{\f}{{\bm f}}

\newcommand{\bA}{{\bm A}}

\newcommand{\bI}{{\bm I}}
\newcommand{\bR}{{\bm R}}

\newcommand{\bW}{{\bm W}}

\newcommand{\bF}{{\bf F}}
\newcommand{\bD}{{\bf D}}

\newcommand{\bS}{{\bf S}}
\newcommand{\bG}{{\bf G}}
\newcommand{\bH}{{\bf H}}
\newcommand{\bX}{{\bf X}}

\newcommand{\bp}{{\bm p}}
\newcommand{\bP}{{\bm P}}

\newcommand{\bh}{{\bf h}}

\newcommand{\bY}{{\bf Y}}

\newcommand{\bzero}{{\bf 0}}

\newcommand{\blambda}{\mbox{\boldmath$\lambda$}}
\newcommand{\bGamma}{\mbox{\boldmath$\Gamma$}}

\newcommand{\bPsi}{\mbox{\boldmath$\Psi$}}

\newcommand{\brho}{\mbox{\boldmath{$\rho$}}}

\newcommand{\expect}{{\mathbb{E}}}
\begin{document}

\raggedbottom

\title{Joint Channel Estimation and Equalization in Massive MIMO Using a Single Pilot Subcarrier}
\author{\normalsize Danilo Lelin Li, Arman Farhang \vspace{-0.7cm}
\thanks{Danilo Lelin Li and Arman Farhang are with the Department of Electronic and Electrical Engineering, Trinity College Dublin, Dublin 2, Ireland
(Email: \{lelinlid,arman.farhang\}@tcd.ie).}}

\maketitle

\begin{abstract}

The focus of this letter is on the reduction of the large pilot overhead in orthogonal frequency division multiplexing (OFDM) based massive multiple-input multiple-output (MIMO) systems. We propose a novel joint channel estimation and equalization technique that requires only one pilot subcarrier, reducing the pilot overhead by orders of magnitude. We take advantage of the coherent bandwidth spanning over multiple subcarrier bands. This allows for a band of subcarriers to be equalized with the channel frequency response (CFR) at a single subcarrier. Subsequently, the detected data symbols are considered as virtual pilots, and their CFRs are updated without additional pilot overhead. Thereafter, the remaining channel estimation and equalization can be performed in a sliding manner. With this approach, we use multiple channel estimates to equalize the data at each subcarrier. This allows us to take advantage of frequency diversity and improve the detection performance. Finally, we corroborate the above claims through extensive numerical analysis, showing the superior performance of our proposed technique compared to conventional methods.

\end{abstract}


\vspace{-0.4cm}
\section{Introduction}
\vspace{-0.1cm}

Massive multiple-input multiple-output (MIMO) has brought substantial improvements to 5G systems and it continues to be a dominant technology in next-generation networks. In particular, spatial diversity and beamforming gains provided by massive MIMO have led to significant improvements in capacity, spectral, and energy efficiency compared to $4$G networks  \cite{7880691}. 
In massive MIMO, multiple users reuse the same time-frequency resources, 
as their signals can be distinguished based on their channel responses.
Therefore, obtaining accurate channel state information (CSI) at each base station (BS) antenna is crucial. Conventional channel estimation methods, e.g. least squares (LS), rely on the transmission of pilot sequences from each user. These methods are favorable due to their low complexity. However, this comes at the expense of spectral efficiency (SE) loss. 

Hence, extensive research has been conducted to reduce the pilot overhead, see \cite{Semiblind,MICED,Superimposed} and the references therein. For instance, (semi) blind techniques emphasize    reducing the pilot length and estimating the channel by solving underdetermined systems of equations using iterative algorithms \cite{Semiblind}. Moreover, superimposed pilots share the same time-frequency resources for both data and pilot transmission. In these techniques, the channel is estimated by treating data as noise. This can completely remove the pilot overhead, however, at the cost of interference and extra computational load \cite{MICED,Superimposed}. Most works on massive MIMO consider a narrow band channel. Therefore, utilizing these channel estimation techniques for orthogonal frequency division multiplexing (OFDM) would require pilots on each subcarrier. This is while it is shown in \cite{L}
that the number of pilot subcarriers needs to be at least equal to the channel length. However, this leads to a large overhead, as the channel length varies from $7\%$ to $25\%$ of the OFDM symbol duration, \cite{3gpp2018nr}. To alleviate this issue, the authors in \cite{L} proposed a technique that can reduce the required pilot subcarriers by $80\%$, however, the technique is limited to sparse channels.

To address the aforementioned limitations of the existing literature, in this letter, we show that only one reference subcarrier is sufficient for both channel estimation and data detection in massive MIMO. It is worth noting that our proposed technique is 
independent of channel sparsity and length. We take advantage of multiple adjacent subcarriers being within the channel coherent bandwidth. Thus, the transmitted data over a band of subcarriers can be detected using only one reference subcarrier's channel estimate. Thereafter, the detected data symbols are considered as virtual pilots to update their corresponding channel estimates. Each updated channel estimate is then used to equalize a new band of subcarriers within the coherent bandwidth. This way, we can obtain multiple estimates of the transmitted data symbols at each subcarrier and average them. With this approach, we take advantage of the frequency diversity apart from the spatial diversity to achieve an improved performance. By repeating this process in a sliding manner, all the data symbols and the CFRs can be estimated, starting from one reference pilot.

Furthermore, we prove that in the large antenna regime, linear combining can be performed using the channel estimate at any subcarrier. The combined signal at a subcarrier will be scaled by a coefficient, which depends on the channel power delay profile (PDP) and frequency spacing between the subcarrier and the CFR used for combining. However, this coefficient can take very small values leading to noise enhancement. To avoid this issue, we define a \textit{depth} factor in our proposed sliding technique that determines the number of subcarriers within each band to be equalized using the same CFR. This ensures that only large values of the aforementioned coefficient are considered. To evaluate the effectiveness of our proposed technique, we compare its bit error rate (BER) and signal-to-interference and noise ratio (SINR) performance with the conventional techniques. For $15$~MHz transmission bandwidth, the existing techniques require over $70$ pilot subcarriers, whereas our proposed technique reduces this overhead to only one pilot subcarrier. Moreover, our proposed technique outperforms the linear equalizers by around $2$~dB, at high signal-to-noise ratios (SNRs), in terms of the BER performance.

{\it Notations}: Matrices, vectors, and scalar quantities are denoted by boldface uppercase, boldface lowercase, and normal
letters, respectively. $[\bA]_{m,n}$
represents the element on row $m$ and column $n$ of $\bA$. ${\rm{tr}}\{\bA \}$ represents the trace of $\bA$. $\bI_M$ and $\textbf{0}_{M\times N}$ are the $M \times M$ identity and $M \times N$ zero matrices, respectively. The superscripts $(\cdot)^{\dagger}$, $(\cdot)^{\rm H}$, $(\cdot)^{\rm T}$ and $(\cdot)^{*}$ indicate pseudo inverse, Hermitian, transpose and conjugate operations, respectively. $|\cdot |$, $(( \cdot ))_{M}$ and $\expect\{\cdot \}$ are the absolute value, modulo-$M$, and expected value operators, respectively. Finally, $\bF_M$ is the normalized $M$-point discrete Fourier transform (DFT) matrix with the elements $[\bF_M]_{m,n}$ = $\frac{1}{\sqrt{M}}e^{\frac{-j2\pi mn}{M}}$, for $m,~n = 0,..., M-1$ and $\f_{M,m}$ is the $m^{\rm th}$ column of $\bF_M$.

\vspace{-0.25cm}
\section{System Model}
\label{sec:sysmod}
\vspace{-0.1cm}

We consider the uplink (UL) of a single-cell massive MIMO system operating in the time-division duplex (TDD) mode. OFDM with $M$ subcarriers  is deployed as the modulation format, with the cyclic prefix (CP) of length $M_{\rm CP}$. $M_{\rm CP}$ is chosen to be larger than the channel length to avoid intersymbol interference. The BS is equipped with $Q$ antennas and serves $K$ single antenna users. The duration of each frame is assumed to be within the channel coherence time, including $N$ number of  OFDM  time symbols in the UL. Hence, the channel remains time-invariant within each frame.

The UL transmission is divided into two phases, training/pilot and data transmission. $N_{\rm{p}}$ and $N_{\rm{d}}$ time slots are allocated to pilot and data transmission, respectively. Thus, a given user $k$ transmits the time-frequency symbols $\bX^{k}=[\bP^{k},\bD^{k}]$, where $\bP^{k} \in \mathbb{C}^{M\times N_{\rm{p}}}$ represents the pilot and $\bD^{k} \in \mathbb{C}^{M\times N_{\rm{d}}}$ the data symbols. Hence, the OFDM transmit signal of user $k$ is obtained as $\bS^{k}= \bA_{\rm{CP}} \bF_M^{\rm H} \bX^{k}$, where $\bA_{\rm{CP}} = [\bG_{\rm{CP}}^{\rm T} , \bI_M ]^{\rm T}$ is the CP addition matrix and the $M_{\rm{CP}} \times M$ matrix $\bG_{\rm{CP}}$ includes the last $M_{\rm{CP}}$ rows of $\bI_M$. 

The signal is transmitted through the channel and undergoes OFDM demodulation. Thus, the received signal from all the users at a given BS antenna $q$ is obtained as 
\begin{align}\label{eq:matR}
    \begin{split}
        \bY_{q}&=\bF_M \bR_{\rm{CP}} \sum^{K-1}_{k=0} \bH_{q,k} \bS^{k} + \bW_{q}\\
        &= \sum^{K-1}_{k=0} \rm{diag}\{ \blambda_{q,k} \} \bX^k + \bW_{q},
    \end{split}
\end{align}
where $\bR_{\rm{CP}}=[\bzero_{M \times M_{\rm{CP}}},\bI_M]$ is the CP removal matrix, $\bH_{q,k}$ denotes the Toeplitz channel matrix realizing the linear convolution, which is formed by  the channel impulse response (CIR) between user $k$ and antenna $q$, i.e., $\bh_{q,k} =[h_{q,k}[0],\ldots,h_{q,k}[L-1]]^{\rm T}$, and $\bW_{q} $ includes the complex additive white Gaussian noise (AWGN) in the frequency domain{, with the variance $\sigma_w^2$, i.e., $[{\bW}_{q}]_{m,n} \sim \mathcal{CN}(0,\sigma_w^2)$}. We assume the CIR between the UEs and the BS antennas to be independent and identically distributed (i.i.d.) complex random variables with the length $L$, $\bh_{q,k} \sim \mathcal{CN}(\bzero_{L\times 1},\boldsymbol{\Sigma}_{k})$ for $q=0,\ldots, Q-1$ and $k=0,\ldots, K-1$. $\boldsymbol{\Sigma}_{k}$ is a diagonal matrix with the diagonal elements formed by the PDP of the channel $\brho_{k} = [\rho_{k}[0],\ldots,\rho_{k}[L-1]]^{\rm{T}}$, where the normalized PDP is considered, i.e., $\sum^{L-1}_{l=0} \rho_{k}[l] = 1$.  As $M_{\rm{CP}} \geq L-1$, $\widetilde{\bH}_{q,k}=\bR_{\rm{CP}} \bH_{q,k} \bA_{\rm{CP}}$ is a circulant matrix, with the first column formed by the zero-padded CIR, i.e., $\widetilde{\bh}_{q,k}=[\bh^{\rm{T}}_{q,k},\bzero^{\rm{T}}_{M-L\times 1}]^{\rm{T}}$. Hence, the channel matrix $\widetilde{\bH}_{q,k}$ can be diagonalized by the DFT and inverse DFT matrices as $\bF_M \widetilde{\bH}_{q,k} \bF_M^{\rm H}=\rm{diag}\{ \blambda_{q,k} \}$, where $\blambda_{q,k} = [\lambda_{q,k}[0],\dots,\lambda_{q,k}[M-1] ]^{\rm{T}}$ represents the CFR with the elements obtained as  
\be\label{eq:lambda}
    \lambda_{q,k}[m]=\sqrt{M} \f_{M,m}^{\rm{T}}\widetilde{\bh}_{q,k}.
\ee

To pave the way toward the derivations in the following sections, we rearrange the received signals into the space-time representation. 
By stacking the received samples at a given subcarrier $m$, i.e., the $m^{\rm{th}}$ row of $\bY_q$ from all the $Q$ receive antennas, we obtain the $Q \times N$ space-time matrix $\overline{\bY}_{m} $. Thus, the input-output relationship for a given subcarrier $m$ across all the antennas can be represented as 
\be\label{eq:freqbin}
    \overline{\bY}_{m}=\overline{\boldsymbol{\Lambda}}_{m} \overline{\bX}_{m} + \overline{\bW}_{m},
\ee
where the elements of the matrices $\overline{\boldsymbol{\Lambda}}_{m}\in \mathbb{C}^{Q\times K}$, $\overline{\bX}_{m}\in \mathbb{C}^{K\times N}$ and $\overline{\bW}_{m} \in \mathbb{C}^{Q \times N}$ are given by $[\overline{\boldsymbol{\Lambda}}_{m}]_{q,k}=\lambda_{q,k}[m]$, $[\overline{\bX}_{m}]_{k,n}=[\bX^{k}]_{m,n}$ and $[\overline{\bW}_{m}]_{q,n}=[{\bW}_{q}]_{m,n}$, respectively. 

Finally, $\overline{\bX}_{m} \triangleq [\overline{\bP}_{m},\overline{\bD}_{m}]$, where $\overline{\bP}_{m}$ and $\overline{\bD}_{m}$ represent the transmitted pilot and data symbols, respectively.
Correspondingly, $\overline{\bY}_{m} \triangleq [\overline{\bY}^{\rm{p}}_{m},\overline{\bY}^{\rm{d}}_{m}]$ and $\overline{\bW}_{m} \triangleq [\overline{\bW}^{\rm{p}}_{m},\overline{\bW}^{\rm{d}}_{m}]$.

\vspace{-0.3cm}
\section{Conventional Channel Estimation and Equalization}
\label{sec:CEandE}
\vspace{-0.1cm}
 
In this section, the widely used linear pilot-based channel estimation and equalization techniques are described. In particular, we consider LS-based channel estimation, and two linear combining techniques, namely MRC and MMSE equalization techniques. At the channel sounding stage, each user transmits a pilot sequence with the length $N_{\rm{p}} \geq K$ over different time slots on a given subcarrier allocated to pilot symbols \cite{7880691}.   The pilot sequence of length $N_{\rm{p}} $ for a given user $k$ at a given subcarrier $m$, $\overline{\bp}^k_m$, is chosen from a pilot book $\overline{\bP}_m=[\overline{\bp}^0_m,\overline{\bp}^1_m,\dots,\overline{\bp}_m^{K-1}]^{\rm T}$ that satisfies the orthogonality property $\overline{\bP}_m\overline{\bP}_m^{{\rm{H}}} = N_{\rm{p}}\bI_K$, \cite{7294693}. We consider a pilot book where each $\overline{\bp}_m^k$ is obtained from $k$ cyclic shifts of the Zadoff-Chu (ZC) root sequence of length $N_{\rm{p}}$. Using \eqref{eq:freqbin}, the received pilots from all the users at subcarrier $m$ can be represented as $  \overline{\bY}_{m}^{\rm{p}}= \overline{\boldsymbol{\Lambda}}_{m} \overline{\bP}_m + \overline{\bW}^{\rm{p}}_{m}$.
Consequently, the channel response for all the users at a given subcarrier $m$ can be estimated as
\be\label{eq:CE}
  \widehat{\boldsymbol{\Lambda}}_{m}= \frac{1}{N_{\rm{p}}}  \overline{\bY}_{m}^{\rm{p}}\overline{\bP}_m^{{\rm{H}}}= \overline{\boldsymbol{\Lambda}}_{m}  + \frac{1}{N_{\rm{p}}}  \overline{\bW}^{\rm{p}}_{m}\overline{\bP}_m^{{\rm{H}}} \triangleq \overline{\boldsymbol{\Lambda}}_{m}  + \widetilde{\bW}_{m},
\ee
where $\widetilde{\bW}_{m}=\frac{1}{N_{\rm{p}}}  \overline{\bW}^{\rm{p}}_{m}\overline{\bP}_m^{{\rm{H}}}$.

To estimate the whole CFR, it is sufficient to allocate $L$ subcarriers as pilots \cite{L}. Let  $\mathcal{I}$ represent the set of subcarrier indices allocated to the pilot, and $\blambda_{q,k}^{\mathcal{I}}$ the vector formed by selecting the $L$ entries of $\blambda_{q,k}$ at the pilot subcarriers, each element of $\blambda_{q,k}^{\mathcal{I}}$ can be estimated from \eqref{eq:CE}, as $\lambda_{q,k}[m]=[\widehat{\boldsymbol{\Lambda}}_{m}]_{q,k}$. From  \eqref{eq:lambda}, the following relation between CFR and CIR holds
\be\label{eq:CFRtoCIR}
    \blambda_{q,k}^{\mathcal{I}}=\sqrt{M} \bF_M^{\mathcal{I}} {\bh}_{q,k},
\ee
where $\bF_M^{\mathcal{I}}$ is a $L\times L$ matrix formed by selecting the $L$ first columns of $\bF_M$ with the rows indexed by $\mathcal{I}$. Hence, the CIR of user $k$ at a given BS antenna $q$ can be estimated by solving \eqref{eq:CFRtoCIR} for ${\bh}_{q,k}$, i.e., ${\bh}_{q,k}=\frac{1}{\sqrt{M}} (\bF_M^{\mathcal{I}})^{-1} \blambda_{q,k}^{\mathcal{I}}$. Finally, after obtaining the CIR, $\widehat{\boldsymbol{\Lambda}}_{m}$ at each subcarrier can be reconstructed with \eqref{eq:lambda}. Therefore, the channel estimation process, \eqref{eq:CE} and \eqref{eq:CFRtoCIR}, requires a minimum of $LK$ pilots.

Using the channel estimates $\widehat{\boldsymbol{\Lambda}}_{m}$, and deploying a linear combining technique such as MRC, the transmit data symbols for all the users can be  estimated as 
\be\label{eq:mf}
\widehat{\bX}^{{\rm{MRC}}}_{m}= \bGamma_m^{-1} \widehat{\boldsymbol{\Lambda}}_{m}^{\rm{H}} \bY^{\rm{d}}_{m} = \bGamma_m^{-1} \widehat{\boldsymbol{\Lambda}}_{m}^{\rm{H}} (\overline{\boldsymbol{\Lambda}}_{m} \bD_m + {\bW}_{m}^{\rm{d}}), 
\ee
{where $\bGamma_m$ is a $K\times K$ diagonal matrix, with the $k^{\rm{th}}$ diagonal element given by $\frac{1}{Q}\sum^{Q-1}_{q=0} |\lambda_{q,k}[m]|^2$, considered solely to normalize the amplitude of the equalizer output}. In the literature, the effect of channel estimation noise is often neglected and the diagonal elements of $\bGamma_m$ are formed from the diagonal elements of $\overline{\boldsymbol{\Lambda}}_{m}^{\rm{H}} \overline{\boldsymbol{\Lambda}}_{m}$, i.e., the norm squared of the channel vector for each user at a given subcarrier $m$. However, in the presence of channel estimation errors, using \eqref{eq:CE}, the normalization factors on the elements of $\overline{\boldsymbol{\Lambda}}_{m}^{\rm{H}} \overline{\boldsymbol{\Lambda}}_{m}$ are obtained as the diagonal elements of
\be
    \widehat{\boldsymbol{\Lambda}}_{m}^{\rm{H}} \widehat{\boldsymbol{\Lambda}}_{m} = \widehat{\boldsymbol{\Lambda}}_{m}^{\rm{H}} \overline{\boldsymbol{\Lambda}}_{m} + \overline{\boldsymbol{\Lambda}}_{m}^{\rm{H}}\widetilde{\bW}_{m} + \frac{1}{N_{\rm{p}}^2}  \overline{\bP}_m(\overline{\bW}^{\rm{p}}_{m})^{\rm{H}} \overline{\bW}^{\rm{p}}_{m}\overline{\bP}_m^{{\rm{H}}}.
\ee
The term $\overline{\boldsymbol{\Lambda}}_{m}^{\rm{H}}\widetilde{\bW}_{m}$ tends to zero in the asymptotic regime, as $Q \to \infty$, since the channel gain and noise are independent. However, the same is not true for $(\overline{\bW}^{\rm{p}}_{m})^{\rm{H}} \overline{\bW}^{\rm{p}}_{m}$. As $\overline{\bW}^{\rm{p}}_{m}$ represents the AWGN, in the asymptotic regime, we have
\be\label{eq:CEnoise}
 \left(\overline{\boldsymbol{\Lambda}}_{m}^{\rm{H}}\widetilde{\bW}_{m} + \frac{1}{N_{\rm{p}}^2}  \overline{\bP}_m(\overline{\bW}^{\rm{p}}_{m})^{\rm{H}} \overline{\bW}^{\rm{p}}_{m}\overline{\bP}_m^{{\rm{H}}}\right) \to \frac{Q\sigma_w^2}{N^2_{\rm{p}}} \overline{\bP}_m\overline{\bP}_m^{{\rm{H}}}.
\ee
Hence, taking a similar approach to \cite{Hamed}, the channel estimation noise can be mitigated by subtracting \eqref{eq:CEnoise} from the normalization factor, i.e., $\bGamma_m$ is formed from the diagonal elements of $(\widehat{\boldsymbol{\Lambda}}_{m}^{\rm{H}} \widehat{\boldsymbol{\Lambda}}_{m}-\frac{Q\sigma_w^2}{N^2_{\rm{p}}} \overline{\bP}_m\overline{\bP}_m^{{\rm{H}}})$.

Due to the channel hardening effect \cite{7880691}, in the asymptotic regime, as $Q \to \infty$, we have 
\begin{equation}\label{eq:fp}
  \frac{\overline{\boldsymbol{\lambda}}_{m,k}^{\rm{H}}\overline{\boldsymbol{\lambda}}_{m,k'}}{Q} \to
    \begin{cases}
      \expect \{ |\overline{\lambda}_{m,k}[q]|^2 \}=1 & \text{if } k=k'\\
      0 & \text{otherwise}
    \end{cases},
\end{equation}
where {$\overline{\boldsymbol{\lambda}}_{m,k}$, of length $Q$,} represents the $k^{\rm{th}}$ column of $\overline{\boldsymbol{\Lambda}}_{m}$. Therefore, as $Q \to \infty$, $\bGamma_m^{-1} \widehat{\boldsymbol{\Lambda}}_{m}^{\rm{H}} \overline{\boldsymbol{\Lambda}}_{m} \to \bI_K$.

In practical systems, the number of BS antennas is limited, and the off-diagonal elements of  $\widehat{\boldsymbol{\Lambda}}_{m}^{\rm{H}} \overline{\boldsymbol{\Lambda}}_{m}$, in \eqref{eq:mf}, lead to a significant amount of  interference. Thus, MMSE combining, 
\be\label{eq:MMSEcomb}
\widehat{\bX}^{{\rm{MMSE}}}_{m}= \boldsymbol{\Phi}_{m}^{\rm{MMSE}} \overline{\bY}^{\rm{d}}_{m},
\ee
where $\boldsymbol{\Phi}_{m}^{\rm{MMSE}} \triangleq (\widehat{\boldsymbol{\Lambda}}_{m}^{\rm{H}} \widehat{\boldsymbol{\Lambda}}_{m}- \frac{Q\sigma_w^2}{N^2_{\rm{p}}} \overline{\bP}_m\overline{\bP}_m^{{\rm{H}}} + \sigma_w^2\bI_{K} )^{-1} \widehat{\boldsymbol{\Lambda}}_{m}^{\rm{H}}$, provides an improved performance  compared to MRC.

\vspace{-0.2cm}
\section{Proposed Channel Estimation and Combining}
\label{sec:PropT}
\vspace{-0.1cm}

The linear combining techniques presented in the previous section take advantage of the channel hardening effect and spatial diversity gains in massive MIMO. However, it requires knowledge of the whole CFR, and therefore, a minimum of $LK$ time-frequency slots are allocated to the pilots for channel estimation. In practical systems, $L$ can take large values, especially as the bandwidth increases. Hence, in this section, we study the possibility of reducing the pilot overhead to a single pilot subcarrier. Therefore, only $K$ time-frequency slots would be necessary for pilot allocation, reducing the training overhead by a factor of $L$.

In current standards, the subcarrier spacing $\Delta f$ is chosen to be much shorter than the coherent bandwidth \cite{etsi38138}. Hence, the channel can be considered flat across the frequency band of multiple adjacent subcarriers. {This suggests that the channel at a single subcarrier can be considered for the equalization of its adjacent subcarriers. Moreover, the correlation between the CFR of adjacent subcarriers can be obtained from the following approximation \cite{hlawatsch2011wireless}}
\be\label{eq:approxalpha}
    |{\alpha}_{\Delta m,k}| \approx \sqrt{1-\left(\frac{\Delta f \Delta m}{F_{\rm{c}}}\right)^2},
\ee
where we define ${\alpha}_{\Delta m,k} \triangleq \expect \{ \overline{\lambda}_{m,k}^*[q] \overline{\lambda}_{(( m+\Delta m))_{M},k}[q]  \}$, $F_{\rm{c}}=\frac{1}{\ell_{\tau}}$ is the coherence bandwidth and $\ell_{\tau}$ is the maximum delay spread of the channel. 

Taking this into account, we consider a reference pilot in  one subcarrier and utilize its channel estimate for equalization of the neighboring subcarriers. The detected data symbols will be considered as virtual pilots, updating their channel estimates and allowing for the equalization of data at their adjacent subcarriers. We repeat this process, in a sliding manner, until all the data and the CFR in the UL frame are estimated. 
Hence, the data $\overline{\bX}_{m }$, at any given subcarrier $m$, can be equalized using \eqref{eq:MMSEcomb} and \eqref{eq:approxalpha} with the CFR at subcarrier $m-1$ as 
\be\label{eq:propMMSE}
\widehat{\bX}^{{\rm{MMSE}}}_{m}= \frac{1}{|{\alpha}_{1,k}|} \boldsymbol{\Phi}_{m-1}^{\rm{MMSE}} \overline{\bY}_{m}.
\ee
In \eqref{eq:propMMSE}, we consider the approximation $(\widehat{\boldsymbol{\Lambda}}_{m})^{\rm{H}} \widehat{\boldsymbol{\Lambda}}_{m} \approx (\widehat{\boldsymbol{\Lambda}}_{m-1})^{\rm{H}} \widehat{\boldsymbol{\Lambda}}_{m-1}${, and the phase of ${\alpha}_{1,k}$ to be negligible}.

After hard decision of the QAM symbols in $\widehat{\bX}^{{\rm{MMSE}}}_{m}$ to obtain $\widehat{\bX}^{{\rm{HD}}}_{m}$, the CFR at subcarrier $m$ is updated as
\be\label{eq:propLS}
    \widehat{\boldsymbol{\Lambda}}_{m}= \overline{\bY}_{m} (\widehat{\bX}^{{\rm{HD}}}_{m})^{\dagger}.
\ee
This channel estimate is then used for equalization of the following subcarrier, $m+1$, using \eqref{eq:propMMSE}. { In particular, since we now consider $\widehat{\bX}_{m}$ as the pilot, the noise mitigation term in \eqref{eq:CEnoise} should be corrected accordingly. That is, the term $\overline{\bP}_m^{\rm{H}}$ in $\bGamma_m$ and $\boldsymbol{\Phi}_{m}^{\rm{MMSE}}$ should be substituted by $(\widehat{\bX}^{{\rm{HD}}}_{m})^{\dagger}$.  Then the procedure is repeated in a sliding manner until the whole frame is equalized. If the detected symbols $\widehat{\bX}^{{\rm{HD}}}_{m}$ are erroneous, \eqref{eq:propLS} provides imperfect channel estimates, and hence, the sliding technique propagates the error to further subcarriers. Therefore, as we will show in the later part of this section, we propose the concept of depth that significantly improves the accuracy of $\widehat{\bX}^{{\rm{HD}}}_{m}$, alleviating the error propagation issue. }

It is worth noting that $\overline{\bX}_{m}$ can be a rank-deficient matrix. Specifically, if the rank of $\overline{\bX}_{m}$ is lower than $K$, \eqref{eq:propLS} fails to retrieve $\widehat{\boldsymbol{\Lambda}}_{m}$. {To solve this issue, we calculate  ${\alpha}_{\Delta m,k}$  for any value of $\Delta m$. This way,  $\widehat{\bX}^{{\rm{MMSE}}}_{m}$ at any subcarrier can be detected with the same reference pilot using \eqref{eq:propMMSE}.}

\begin{proposition}\label{prop:1}
In the asymptotic regime, it is possible to perform MRC with the channel estimates on a single subcarrier. In other words, the data at any given subcarrier $m$ can be equalized with the CFR at any other subcarrier $m'$.
\end{proposition}

\begin{proof}

We detect the data $\overline{\bX}_{m}$, at any given subcarrier $m$,  by performing MRC with the channel estimates at subcarrier  $m'$
\be\label{eq:propmrc}
\widehat{\bX}^{{\rm{MRC}}}_{m}= \bGamma_{m'}^{-1} \boldsymbol{\Lambda}_{m'}^{\rm{H}} \bY_{m}^{\rm{d}}. 
\ee

Considering \eqref{eq:fp}  in the asymptotic regime, as $Q \to \infty$, 
\begin{equation}\label{eq:fpdetam}
  \frac{\boldsymbol{\overline{\lambda}}_{m',k'}^{\rm{H}}\boldsymbol{\overline{\lambda}}_{m,k}}{Q} \to
    \begin{cases}
      {\alpha}_{\Delta m,k} & \text{if } k=k'\\
      0 & \text{otherwise}
    \end{cases},    
\end{equation}
where $\Delta m=m-m'$. Consequently, in \eqref{eq:propmrc}, $\bGamma_{m'}^{-1} \boldsymbol{\Lambda}_{m'}^{\rm{H}} \boldsymbol{\Lambda}_{m} \to \bPsi_{\Delta m}$, where $\bPsi_{\Delta m}$ is a diagonal matrix with the diagonal elements formed by the vector $[\alpha_{\Delta m,0}, \dots , \alpha_{\Delta m, K-1}]$.

To obtain the {exact}  value of $\alpha_{\Delta m,k}$, we substitute the CFRs with the expression given in \eqref{eq:lambda}. This yields a quadratic form, with the expectation defined as \cite{mathai1992quadratic}
\be\label{eq:trace}
    \expect \{ \widetilde{\bh}_{q,k}^{\rm H}(M \f_{M,m}^{*} \f_{M,(( m+\Delta m))_{M}}^{\rm{T}})\widetilde{\bh}_{q,k}\}= {\rm{tr}}\{\bA \widetilde{\boldsymbol{\Sigma}}_{k}\} + \boldsymbol{\mu}'\bA\boldsymbol{\mu},
\ee
where $\boldsymbol{\mu}$ and $\widetilde{\boldsymbol{\Sigma}}_{k}$ represent the mean and covariance matrices of $\widetilde{\bh}_{q,k}$, respectively, and  $\bA = M\f_{M,m}^{*} \f_{M,(( m+\Delta m))_{M}}^{\rm{T}}$. With the parameters presented in section~\ref{sec:sysmod}, it is given that the channel is zero mean, i.e.,  $\boldsymbol{\mu} = \bzero_{M\times 1}$, and $\widetilde{\boldsymbol{\Sigma}}_{k}$ is a diagonal matrix with the diagonal elements formed by the vector $\widetilde{\brho}_{k}=[{\brho}_{k}^{\rm{T}},\bzero^{\rm{T}}_{M-L\times 1}]^{\rm{T}}$. Hence, \eqref{eq:trace} reduces to
\be\label{eq:lambdadelta}
   \alpha_{\Delta m,k}={\rm{tr}}\{\bA \widetilde{\boldsymbol{\Sigma}}_{k}\} = \sqrt{M} \f_{M,(( \Delta m))_{M}}^{{\rm T}}\widetilde{\brho}_k.
\ee

Based on this result, as long as the PDP is known, equalization can be performed using \eqref{eq:propmrc} followed by scaling the MRC output for each subcarrier by $\frac{1}{\alpha_{\Delta m,k}}$. 
\qed\vspace{-0.1cm}

\end{proof} 

\begin{figure*}[!ht]
\begin{multicols}{3}
\centering
   \includegraphics[scale=0.39]{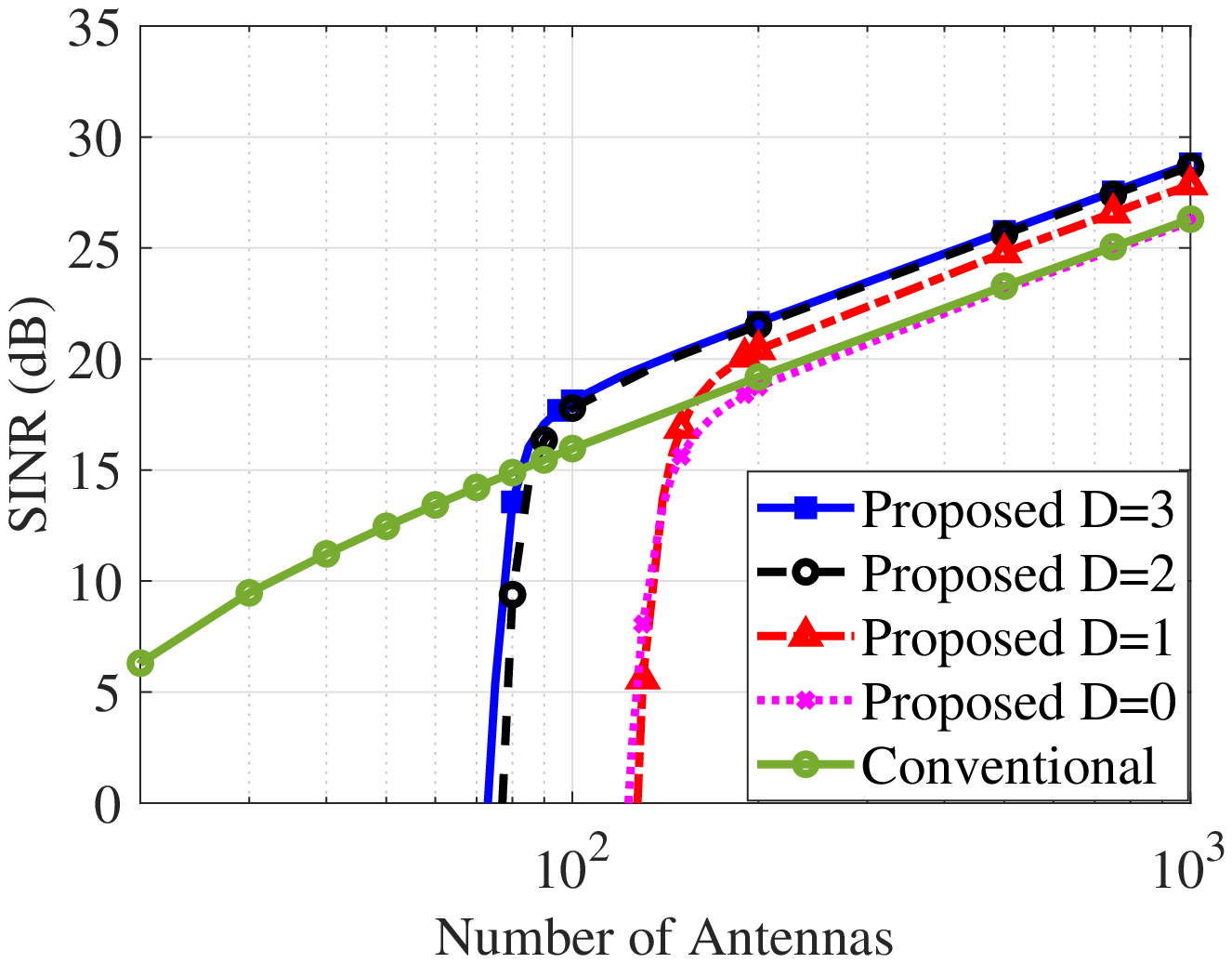}\par \vspace{-0.3cm}\caption{SINR performance as a function of the number of BS antennas.} \label{fig:SINR}
   \includegraphics[scale=0.39]{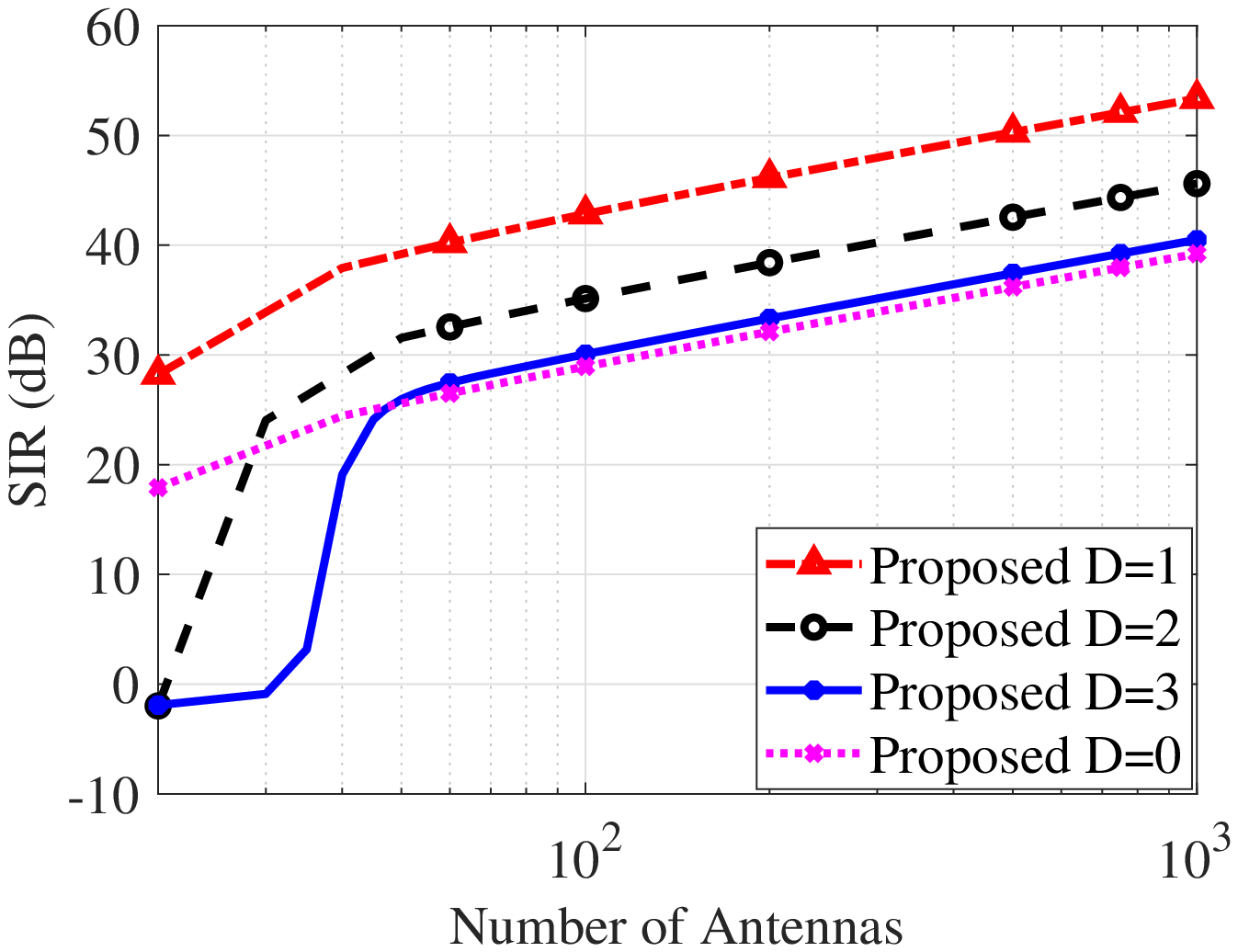}\par \vspace{-0.3cm}\caption{SIR performance as a function of the number of BS antennas.} \label{fig:SINR2}
    \includegraphics[scale=0.39]{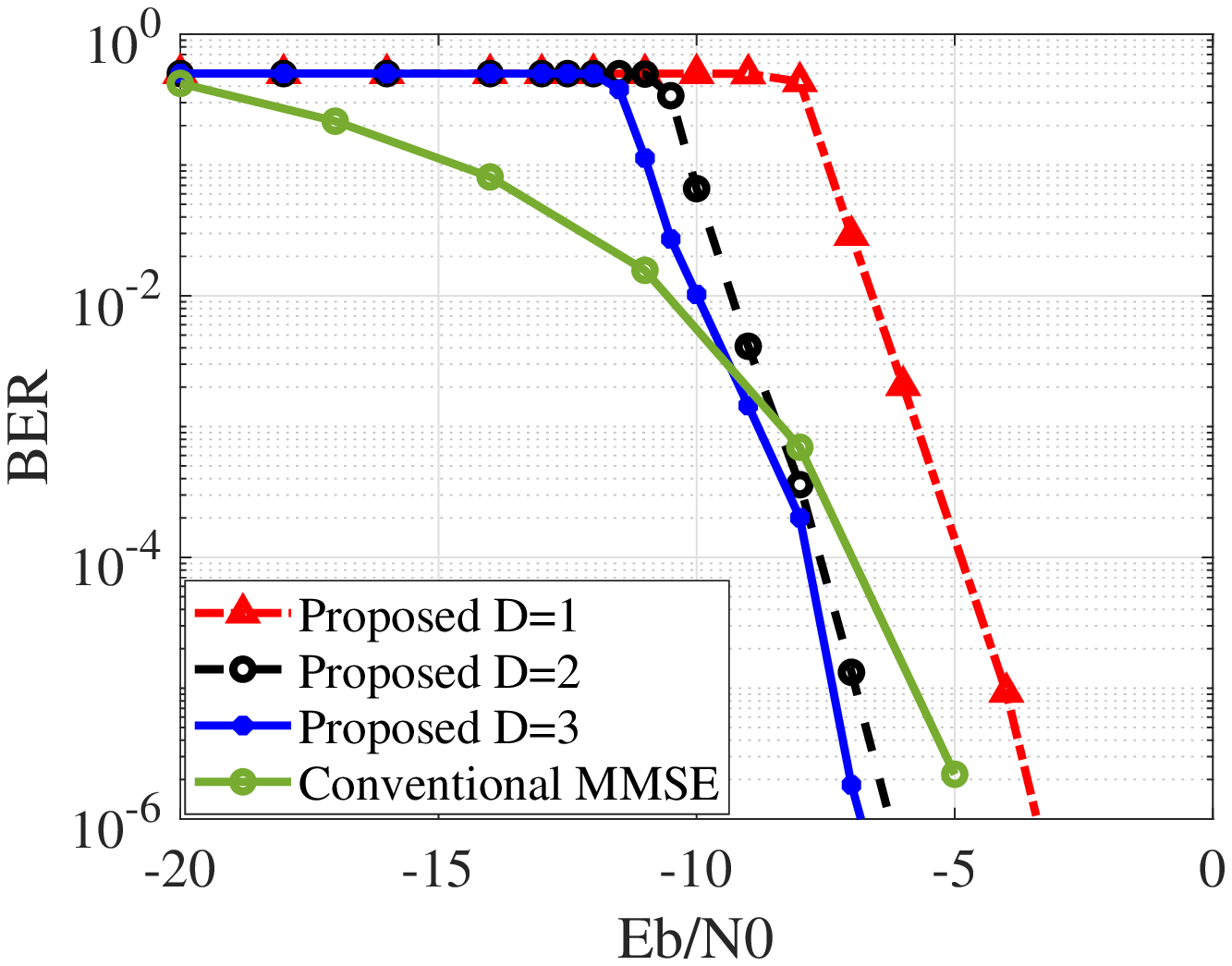}\par \vspace{-0.3cm}\caption{BER performance for our proposed technique and MMSE combining.} \label{fig:BER}
\end{multicols}
\vspace{-0.3cm}
\end{figure*}

While \textit{Proposition}~\ref{prop:1} is sufficient in the asymptotic regime, in practical systems, the aforementioned scaling by $\frac{1}{\alpha_{\Delta m,k}}$ may lead to noise enhancement and, consequently, a performance loss. {For this reason, we rely on our proposed sliding equalization technique as $\alpha_{\Delta m,k}$ is larger for the smaller values of $|\Delta m|$. When $\overline{\bX}_{m}$ is rank-deficient, $\widehat{\boldsymbol{\Lambda}}_{m}$ cannot be estimated. Thus, we resort to utilizing the previously estimated channels on the closest subcarrier. Furthermore, if the PDP is known at the receiver, the scaling term ${\alpha}_{\Delta m,k}$ can be calculated from \eqref{eq:lambdadelta}, instead of using the approximation from \eqref{eq:approxalpha}.}

When the coherence bandwidth is much larger than the subcarrier spacing, the correlation ${\alpha}_{\Delta m,k}$ takes large values for multiple values of $\Delta m$. This sparks the idea of taking advantage of the frequency diversity in addition to the spatial diversity. To this end, improved SINR performance can be achieved by performing equalization for each subcarrier multiple times, each with the channel of a different subcarrier within the coherence bandwidth. This improvement substantially reduces the error propagation issue from \eqref{eq:propLS}.

While the channels from all subcarriers could be considered for the averaging, the channels with a small correlation would bring minimal improvement to the output. Since, the value of ${\alpha}_{\Delta m,k}$ decreases as we increase the value of $|\Delta m|$, only channels of subcarriers within a range are considered. Hence, we call the maximum range considered as the depth,  $D$.  

Ideally, the depth would consider all adjacent subcarriers, at higher and lower indices. However, since the channels are updated in a sliding manner, only the channel on one side of the subcarriers is available at the time of equalization. Therefore, the final output of our proposed technique is achieved with two steps. We perform the first step of the proposed technique by sliding from lower to upper subcarrier indices until the whole frame is equalized. {For the second step, the procedure is realized again, however, now sliding from higher to lower subcarrier indices. It is worth noting that both steps can be realized in parallel as the steps are independent of each other.} The output of both steps are then averaged to obtain the final output of our proposed technique. We distinguish each step with the direction variable $\xi \in [-1,1]$ representing steps one and two, respectively. Each step can be represented as
\be
\widehat{\bX}_{m,\xi}= \frac{1}{D} \sum^{D}_{\Delta m=1} \bPsi_{\xi\Delta m}^{-1}\boldsymbol{\Phi}_{m-\xi\Delta m}^{\rm{MMSE}} \overline{\bY}_{m}. 
\ee
The final output is then obtained as $\widehat{\bX}_{m}=\frac{1}{2}(\widehat{\bX}_{m,-1}+\widehat{\bX}_{m,1})$.
Our proposed technique is summarized in Algorithm \ref{alg:three}.

\setlength{\textfloatsep}{0.5cm}
\begin{algorithm}[t]
\begin{algorithmic}[1]
\caption{Proposed sliding technique with reference pilot at subcarrier index $i$.}
\label{alg:three}

\STATE \textbf{Initialize: }Estimate channel at reference subcarrier (Pilot), obtaining $\widehat{\boldsymbol{\Lambda}}_{i}$ using \eqref{eq:CE}.

\FOR{$\xi=[-1,1]$}
\STATE  $m'=i$ \COMMENT{Define the anchor for closest available CFR}
\FOR{ $m_{\rm{dir}}=1$ to $M-1$}
\STATE $m=((i+\xi m_{\rm{dir}}))_M$ 
\STATE Initialize $\widehat{\bX}_{m,\xi}=\bzero_{K \times N}$ 
\STATE  $c=0$ \COMMENT{Define total channels used for equalization}
\FOR{ $\Delta m =1$ to $D$}
    
    \IF{$\widehat{\boldsymbol{\Lambda}}_{((m-\xi\Delta m))_M}$ was estimated}
        \STATE $\widehat{\bX}_{m,\xi}= \widehat{\bX}_{m,\xi}+\bPsi_{\xi\Delta m}^{-1}\boldsymbol{\Phi}_{((m-\xi\Delta m))_M}^{\rm{MMSE}} \overline{\bY}_{m}$
        \STATE $c = c+1$
    \ENDIF
\ENDFOR     
\IF{$c == 0$}
    \STATE $\widehat{\bX}_{m,\xi}= \bPsi_{m-m'}^{-1}\boldsymbol{\Phi}_{m'}^{\rm{MMSE}} \overline{\bY}_{m}$
\ELSE
    \STATE $\widehat{\bX}_{m,\xi}=\frac{\widehat{\bX}_{m,\xi}}{c} $
\ENDIF

\STATE    $\widehat{\bX}^{{\rm{HD}}}_{m,\xi}$ is obtained from hard decision of $\widehat{\bX}_{m,\xi}$
  
\IF{ $\widehat{\bX}^{{\rm{HD}}}_{m,\xi}$ is invertible}
       
\STATE $ \widehat{\boldsymbol{\Lambda}}_{m}= \overline{\bY}_{m} (\widehat{\bX}^{{\rm{HD}}}_{m})^{\dagger}$ and update $m'=m$
\ELSE
\STATE $ \widehat{\boldsymbol{\Lambda}}_{m}$ is not estimated and $m'$ is not updated
\ENDIF        
        
\ENDFOR     
\ENDFOR

\STATE Obtain $\widehat{\bX}_{m}$ by averaging the two sets of result

\end{algorithmic}
\end{algorithm}

In the following section, we numerically evaluate the performance of our proposed technique. We compare our solution with the conventional LS-based channel estimation and MMSE combining. We show that our proposed channel estimation and equalization method, using only $K$ pilots, outperforms the conventional method, which requires $LK$ pilots.

\vspace{-0.2cm}
\section{Numerical Results}\label{sec:NR}
\vspace{-0.1cm}

In this section, we evaluate the efficacy of our proposed technique and compare it with the conventional LS-based channel estimation and MMSE equalization. In our simulations, we consider the UL transmission of a single frame with $16$-QAM and OFDM modulation, $M=1024$ subcarriers and $N=14$ time symbols. The BS serves $K=7$ users, where $N_{\rm{p}} =7$ time slots are allocated to each user for channel estimation and $ N_{\rm{d}} = 7$ for data transmission. We use the Extended Typical Urban model (ETU) channel model, \cite{lte2009evolved}, which has a coherent bandwidth of $F_{\rm{c}} \approx 200$~kHz. Considering subcarrier spacing of $\Delta f=15$~kHz and transmission bandwidth of $15$~MHz, the channel length is $L=77$. Hence, for the conventional LS-based channel estimators, a total of $LN_{\rm{p}} = 539$ time-frequency resources are allocated for pilot transmission while our proposed technique requires only $N_{\rm{p}} =7$ time-frequency resources. Perfect knowledge of the PDP for each user at the BS is assumed.

In Fig. \ref{fig:SINR}, we evaluate the SINR performance of our proposed technique versus the number of BS antennas for $D=0, 1,2 \text{ and } 3$ and the input SNR of $0$~dB. When $D=0$, we only consider one sliding direction where one estimate of the data symbol at each subcarrier is obtained. The results in Fig. \ref{fig:SINR} show that our proposed technique can effectively average out noise and multiuser interference as the number of BS antennas grows large. Furthermore, when $D=0$, our proposed technique can achieve very close performance to that of the conventional MMSE combiner. This is while for $D=1,2$ and $3$, it leads to the SINR performance gain of around $2$~dB compared to the conventional MMSE combining. The SINR performance of  our proposed technique becomes linear only after deploying a minimum number of BS antennas. This is due to error propagation at the channel estimation stage in \eqref{eq:propLS}. This highlights the benefit of increasing depth as it greatly improves the SINR performance, especially when considering fewer BS antennas or lower input SNR. Furthermore, the frequency diversity predominantly averages out the channel estimation noise in \eqref{eq:CE}. However, as it is shown in Fig.~\ref{fig:SINR2}, in the absence of noise, increasing depth adversely affects the signal-to-interference ratio (SIR). This is due to the approximation in \eqref{eq:propMMSE}. The conventional LS channel estimation and MMSE combining lead to infinite SIR in the absence of noise and thus, cannot be included in the results of Fig.~\ref{fig:SINR2}.

Finally, in Fig. \ref{fig:BER}, we analyze the BER performance of our proposed technique when $Q=200$ BS antennas are deployed. The proposed technique with $D=3$ leads to $2$~dB performance gain at large Eb/N$0$s compared to the conventional method. These results show that increasing $D$ from $1$ to $3$ provides almost $4$~dB performance gain.

\section{~~Conclusion} \label{sec:conc}

In this letter, we proposed a sliding joint channel estimation and equalization technique that significantly reduces the pilot overhead in massive MIMO. With the CFR estimates on a single reference subcarrier, we showed that the data symbols on the adjacent subcarriers can be detected with linear combining. The detected data symbols are then utilized as virtual pilots to update their corresponding channel estimates. The updated CFRs are then used for equalization of the remaining data symbols, in a sliding manner. With this approach, we obtain multiple estimates of the transmitted data symbols at each subcarrier that are averaged to provide additional frequency diversity to the spatial diversity gains of massive MIMO. Through extensive numerical analysis, we showed that our proposed technique achieves improved SINR and BER performance compared to the conventional linear combiners. 

\bibliographystyle{IEEEtran}
\vspace{-0.2cm}
\bibliography{main}
\vspace{-0.1cm}

\end{document}